\pdfoutput=1

\documentclass[11pt]{article}

\usepackage{fullpage}
\usepackage{amsmath}
\usepackage{accents}
\usepackage{amsthm}
\usepackage{amssymb}
\usepackage[hidelinks]{hyperref}
\usepackage{cleveref}
\usepackage{textcomp}
\usepackage{enumerate}
\usepackage{graphicx}
\usepackage{caption}
\usepackage{subcaption}
\usepackage{float}
\usepackage{verbatim}

\newcommand{\mbb}[1]{\mathbb{#1}}
\newcommand{\mbf}[1]{\mathbf{#1}}

\newcommand{\tr}{\textup{tr}}

\newcommand{\mc}{\mathcal}

\newcommand{\diff}{\,\mathrm{d}}
\newcommand{\Var} {\textup{Var}}
\renewcommand{\det}{\textup{det}}
\renewcommand{\Re}{\textup{Re}}

\numberwithin{equation}{section}

\newtheorem{theorem}{Theorem}
\newtheorem{lemma}{Lemma}[section]

\numberwithin{theorem}{section}
\numberwithin{proposition}{section}
\numberwithin{corollary}{section}

\theoremstyle{remark}

\theoremstyle{definition}
\newtheorem{definition}{Definition}[section]

\title{Mesoscopic Linear Statistics for Two Ensembles of Quantum Graphs}
\author{Anna Maltsev and Mohammed Osman}
\date{\small Queen Mary University of London}

\begin{document}

\maketitle
\begin{abstract}
We study mesoscopic linear spectral statistics for two ensembles of random quantum graphs. These are defined by a discrete graph $G$ and a unitary-matrix-valued function $U(k)$ indexed by directed edges of $G$. The matrix function $U(k)$ is constructed from unitary matrices $U^{(v)}$ indexed by the neighbours of each vertex $v$. The first ensemble is obtained by sampling the underlying discrete graph uniformly from the set of $d$-regular graphs. The second ensemble is obtained by sampling $U^{(v)}$ uniformly from the Haar measure, independently for each vertex. We prove that the variance of a linear spectral statistic in the large graph limit on polynomial mesoscopic scales coincides with that of the Gaussian Orthogonal/Unitary Ensemble.
\end{abstract}

\section{Introduction}

Quantum graphs are metric graphs equipped with differential operators and vertex conditions, typically the Schr\"odinger operator with Kirchhoff/Dirichlet vertex conditions. They arise as effective models for wave propagation in thin networks and mesoscopic systems, where motion is one-dimensional along edges with the boundary conditions at vertices providing local and topological structure, interpolating between discrete graph models and continuum systems. They have found a host of applications such as the study of quantum wires, photonic crystals, and waveguide problems (for a review see \cite{kuchment2008quantum} and references therein).

From the perspective of quantum chaos, quantum graphs provide a simplified setting in which to study spectral statistics and their relation to underlying classical dynamics. The Bohigas--Giannoni--Schmit (BGS) conjecture \cite{bohigas_characterization_1984} predicts that for quantum systems whose associated classical dynamics is chaotic, the local spectral correlations of the quantum spectrum are universal and coincide with those of Gaussian Unitary or Gaussian Orthogonal random matrix ensembles (abbreviated GUE, GOE respectively), whereas integrable dynamics should lead to Poisson statistics consistent with Berry-Tabor conjecture \cite{berry_tabor1977}. 

For quantum graphs this picture is supported heuristically, numerically, and experimentally: periodic-orbit expansions and semiclassical arguments yield random-matrix form factors and correlation functions in many quantum graph models \cite{kottos_smilansky1997,kottos_smilansky1999, bolte_trace_2009} and microwave-network laboratory experiments provide some evidence of random matrix statistics in integrated nearest-neighbor spacing and spectral rigidity in real-world physics \cite{hul_et_al2004}. On the rigorous side, two key structural inputs can make quantum graphs amenable to a combinatorial analysis: an exact trace formula expressing spectral quantities in terms of closed walks on the underlying graph \cite{bolte_trace_2009} and an explicit classification of self-adjoint vertex conditions \cite{kostrykin_schrader1999}. Nevertheless, a general proof of the BGS conjecture for deterministic quantum graphs remains open.

One way to distinguish between random matrix spectral statistics and Poisson statistics is by the variance of linear spectral statistics. A linear spectral statistic (LSS) for a given matrix or operator $M$ with eigenvalues $(\lambda_j)$ and a test function $h$ is defined as 
\begin{equation}
	\tr h(M)=\sum_j h(\lambda_j).
\end{equation}
For example, choosing $h=\chi_{I}$ to be the characteristic function of an interval $I$ we obtain the number $N(I)$ of eigenvalues in $I$. If the $\lambda_{j}$ form a Poisson process, the variance of $N(I)$ scales like the length of the interval: $\Var N(I)\sim |I|$. By contrast, for eigenvalues of the GOE/GUE, the variance of $N(I)$ grows logarithmically in $|I|$ due to strong correlations.

Linear spectral statistics are a common object of study in random matrix theory (RMT). For Wigner and related ensembles, the LSS satisfies Gaussian central limit theorems (CLTs) with an explicit variance and these results extend to mesoscopic scales when the test function decays rapidly outside an interval of size $\eta$, where $N^{-1}\ll\eta\ll1$ \cite{landon2022almost}. A universal feature arising in many RMT settings is the $H^{1/2}$–variance of the LSS
\begin{equation}\label{eq:RMTvar}
\Var\!\left(\sum_j h\!\left(\tfrac{\lambda_j - E}{\eta}\right)\right)
\sim 2 \int_{\mathbb{R}} |x|\,\bigl|\widehat{h}(x)\bigr|^2 \, \mathrm{d}x,
\end{equation}
for certain classes of test functions $h$ with Fourier transform $\hat{h}$ and $\eta$ satisfying $N^{-1+\varepsilon}\ll \eta\ll 1$ \cite{johansson1998, lytova_pastur2009, shcherbina2013, lodhia_simm2015}. 

In contrast to the extensive literature on universality of mesoscopic fluctuations of LSS in RMT, much less is known for quantum graphs. There is a large body of work on local spectral statistics, form factors, and number variance for specific ensembles of quantum graphs, often using periodic‑orbit expansions and semiclassical heuristics \cite{kottos_smilansky1997,kottos_smilansky1999, bolte_trace_2009}. However, rigorous results identifying the variance of mesoscopic linear statistics of the spectrum of a quantum graph, as a functional of a general test function $h$, are currently only available for the special case of star graphs \cite{BerkolaikoKeating1999, Chabanol2007}, which deviate from Poisson and RMT statistics.

In this work, we offer a rigorous computation of the variance of LSS for two ensembles of random quantum graphs. We take a more abstract perspective and identify a quantum graph with a metric graph for which each vertex carries a unitary matrix indexed by its neighbours (this is similar to the concept of a quantum walk, see e.g. \cite{AharonovAmbainisKempeVazirani2001QuantumWalksOnGraphs, Joye2024UnitaryOpenScatteringQuantumWalksOnGraphs}). For the first ensemble we fix the vertex unitaries and randomise the underlying discrete graph, while for the second we fix the graph and randomise the vertex unitaries.  The latter model has been introduced by Kottos-Schanz \cite{kottos2001quantum}, who suggest how to reduce the trace formula for the form factor to a combinatorial problem and show numerically that the form factor coincides with RMT predictions. In our main results we recover the GOE/GUE variance on mesoscopic scales for these two models in the large graph limit.

The first ensemble, in which the underlying graph is a random $d$-regular graph, can be thought of as a quantum graph analogue of a random compact hyperbolic surface. For the Weil--Petersson model of such random surfaces, Rudnick \cite{rudnick_goe_2023} shows that the variance of the LSS in the large genus limit converges to that of the GOE. The genus $g$ plays the role of the number of vertices $N$ of the graph. In Rudnick's paper, the GOE result is obtained after taking a double limit, sending $g\to\infty$ before $\eta\to0$. In the case of quantum graphs, we can obtain better bounds on expected numbers of certain types of closed walks which allow us to take $\eta$ dependent on $N$.

%The starting point of our methodology is Theorem 5.3 in \cite{bolte_trace_2009}, which connects spectral fluctuations to sums over closed paths. In this theorem, the authors offer a trace formula for compact metric graphs with general local self-adjoint boundary conditions (including Robin-type conditions, where the scattering depends on the wave number $k$). For a test function $h$ and applied to the ensemble of this paper, the linear spectral statistic is given by the weighted spectral sum $\sum_{n\ge 0} g_n\,h(k_n)$ with a Weyl term $L\widehat{h}(0)$ and a periodic-orbit contribution written as a double sum over topological lengths $\ell$ and periodic orbits $p\in\mathcal P_\ell$. The Weyl term vanishes in the centering and the fluctuations are driven solely by the periodic orbit contribution. Squaring this and using that via a Weingarten calculus a mixed moment of entries is nonzero only if the ``unstarred" indices match the ``starred” indices up to permutation, the sum reduces to one over pairs of paths that traverse the same set of edges with the same
%multiplicities. Then we show that only simple cycles contribute to the variance at mesoscopic scales and compute their contribution explicitly.
%The assumptions of uniform boundedness and convergence of Ces\`aro averages of the edge lengths ensure that the periodic–orbit sum can be controlled uniformly in $N$ and that no single very long orbit dominates the variance. 

The paper is organised as follows: in Section \ref{s:Notations} we introduce our model and our main results; in Section \ref{s:traceFormula} we recall the trace formula of Bolte--Endres \cite{bolte_trace_2009}, which is the starting point for the proofs of both main theorems, which are given in Sections \ref{s:Proof1} and \ref{s:Proof2}; in Section \ref{s:Remarks} we make some further remarks and present some numerical results.

\section{Main Results}\label{s:Notations}

Let $G$ be a graph with vertex set $V$ and edge set $E$. Let $d_{i}$ denote the degree of vertex $i\in V$. We assign a unitary $U^{(i)}\in \textup{U}(d_{i})$ to each vertex $i\in V$ and a length $l_{ij}=l_{ji}>0$ to each edge ${i,j}\in E$. Define the following unitary matrix-valued function
\begin{align}
    U_{(i,j),(k,l)}(\lambda)&=\delta_{j,k}U^{(j)}_{i,l}e^{i\lambda l_{ij}}.
\end{align}
We are interested in the discrete unitary evolution $\mbb{C}^{2|E|}\to\mbb{C}^{2|E|}$ given by
\begin{align}
    \psi_{n+1}&=U(\lambda)\psi_{n}.
\end{align}
For certain values $\lambda_{1}\leq\lambda_{2}\leq\cdots$, there exist stationary solutions \[\psi_{j}=U(\lambda_{j})\psi_{j};\] we will call these values quantum graph eigenvalues and the associated solutions eigenfunctions. The terminology comes from the problem of solving
\begin{align}
    -\Delta f&=\lambda^{2} f,\quad f\in\mc{D}\subset\mc{H},
\end{align}
where the domain $\mc{D}$ is determined by conditions on the values of $f$ and $f'$ at the vertices. Certain choices of $U$ correspond to a choice of domain $\mc{D}$ for which $-\Delta$ is a self-adjoint operator. In this case the $\lambda_{j}$ and $\psi_{j}$ are precisely the eigenvalues and eigenfunctions of $-\Delta$. For example, choosing
\begin{align}
    U^{(j)}_{i,l}&=\frac{2}{d_{j}}-\delta_{il}
\end{align}
corresponds to Neumann--Kirchhoff conditions at the vertices, i.e. the domain consists of functions whose sum of derivatives pointing at each vertex vanishes.

In view of this discussion we are led to the following definition.
\begin{definition}
Let $G=(V,E)$ be a metric graph with vertex set $V$, edge set $E$ and edge lengths $\mbf{l}\in\mbb{R}_{+}^{|E|}$. For each vertex $v\in V$, let $U^{(v)}\in \textup{U}(d_{v})$ be a unitary matrix indexed by the $d_{v}$ neighbours of $v$. Let $U:\mbb{R}\to \textup{U}(2|E|)$ be the unitary matrix-valued function 
\begin{align}
    U_{(i,j),(k,l)}(\lambda)&=\delta_{j,k}U^{(j)}_{i,l}e^{i\lambda l_{ij}}.
\end{align}
A \textbf{quantum graph} $\mc{G}$ is the triple $(G,\mbf{l},U)$. We call the elements of the discrete set
\begin{align}
    \sigma(\mc{G})&:=\{k:\det(1-U(k))=0\}
\end{align}
\textbf{eigenvalues} of the quantum graph, and the corresponding \textbf{eigenfunctions} are the vectors in the set
\begin{align}
    \bigcup_{k\in\sigma(\mc{G})}\bigl\{\textup{ker}(1-U(k))\bigr\}.
\end{align}
\end{definition}

We will consider two models of random quantum graphs, one obtained by fixing $U$ and randomising $G$, the other obtained by fixing $G$ and randomising $U$. First we recall the special ``equi-transmitting" unitaries of \cite{harrison_quantum_2007}:
\begin{definition}[Definition 1.1 in \cite{harrison_quantum_2007}]
A unitary $U\in U(d)$ is \textbf{equi-transmitting} if 
\begin{align}
    |U_{ij}|&=\frac{1-\delta_{ij}}{\sqrt{d-1}}.
\end{align}
\end{definition}
In other words, the diagonal entries vanish and the modulus of all other entries is equal to the same constant. It is shown in \cite[Proposition 3.2, Corollary 3.5]{harrison_quantum_2007} that symmetric equi-transmitting matrices exist for infinite sequences of dimension $d$.

\begin{definition}[Model 1]
Let $d\geq3$ and $G_{1}$ be sampled uniformly from $\mc{G}_{N,d}$, the set of $d$-regular graphs with $N$ vertices. Let $\mu$ be an absolutely continuous measure with compact support in $\mbb{R}_{+}$ and $\mbf{l}_{1}\in\mbb{R}^{|E(G)|}_{+}$ be a vector of independent lengths sampled from $\mu$. Let $U^{(i)}$ be symmetric and equi-transmitting for each $i=1,...,|V(G)|$. We set $\mc{G}_{1}=(G_{1},\mbf{l}_{1},U_{1})$ to be the associated random quantum graph.
\end{definition}

Implicit in the above definition is the condition that $d$ belongs to the set of dimensions for which symmetric, equi-transmitting matrices exist.

\begin{definition}[Model 2]
Let $G_{2}=K_{N+1}$ be the complete graph on $N+1$ vertices and $\mbf{l}\in\mbb{R}^{|E(G)|}_{+}$. For each vertex $v\in[N+1]$, fix an ordering $f_{v}:[N+1]\setminus\{v\}\to[N]$ of its neighbours. Let $U^{(v)}\in \textup{U}(N),\,v\in V(G)$ be independent unitaries sampled uniformly from the Haar measure. We set $\mc{G}_{2}=(G_{2},\mbf{l}_{2},U_{2})$ to be the associated random quantum graph.
\end{definition}

For $h\in C^{\infty}(\mbb{R}_{+})$ and $\lambda,\eta>0$, let
\begin{align}
    h_{\lambda,\eta}(x)&=h\left(\frac{x-\lambda}{\eta}\right).
\end{align}
The Fourier transform is
\begin{align}
    \hat{h}_{\lambda,\eta}(k)&=\eta e^{ik\lambda}\hat{h}(k\eta).
\end{align}
Let
\begin{align}
    \mc{L}_{\lambda,\eta,h}(\mc{G})&=\sum_{j\geq1}h_{\lambda,\eta}(k_{j})
\end{align}
denote a linear statistic of the quantum graph eigenvalues. We will always work with the class of functions $h$ that have a compactly supported Fourier transform.

\begin{theorem}\label{thm1}
There is a constant $c_{\mu}>0$ such that, for $\lambda>c_{\mu}\sqrt{d-1}$ and $\eta>0$, we have
\begin{align}
    \mbb{E}\bigl|\mc{L}_{\lambda,\eta,h}(\mc{G}_{1})-\mbb{E}\mc{L}_{\lambda,\eta,h}(\mc{G}_{1})\bigr|^{2}&=4\int_{0}^{\infty}x|\hat{h}(x)|^{2}\diff x+O\Biggl(\eta^{1/2}+\frac{1}{N\eta^{3}}\Biggr).\label{eq:thm1}
\end{align}
\end{theorem}

\begin{theorem}\label{thm2}
Let $\lambda,\,\eta>0$. Assume that $\sup_{N}\|\mbf{l}_{2}\|_{\infty}<\infty$ and the Ces\`{a}ro sums of $\mbf{l}_{2}$ converge with rate $\phi(N)$. Then we have
\begin{align}
    \mbb{E}\bigl|\mc{L}_{\lambda,\eta,h}(\mc{G}_{2})-\mbb{E}\mc{L}_{\lambda,\eta,h}(\mc{G}_{2})\bigr|^{2}&=2\int_{0}^{\infty}x|\hat{h}(x)|^{2}\diff x+O\Biggl(\eta^{1/2}+\frac{1}{N\eta^{2}}+\phi(N)\Biggr).\label{eq:thm2}
\end{align}
\end{theorem}
The leading terms in the right hand sides of \eqref{eq:thm1} and \eqref{eq:thm2} are precisely the variance of mesoscopic linear statistics of the Gaussian Orthogonal and Unitary Ensembles (GOE and GUE) respectively.

Note that by Weyl's law, the density of eigenvalues is of the same order as the total length of the graph, i.e. the sum of all edge lengths. Thus, for bounded edge lengths, the microscopic regime is $\eta\simeq |E|^{-1}$, where we recall that $|E|$ is the number of edges. In Model 1, the degree $d$ is fixed so $|E|\simeq N$, but in Model 2 $|E|\simeq N^{2}$. Therefore Theorem \ref{thm2} is weaker than Theorem \ref{thm1} in the sense that it holds for $\eta\gg |E|^{-1/4}$ compared to $\eta\gg |E|^{-1/3}$. Ultimately, if both models are in the setting of the BGS conjecture, then both results should hold for $\eta\gg |E|^{-1}$.

\section{Trace Formula}\label{s:traceFormula}

The starting point for both Theorems \ref{thm1} and \ref{thm2} is the trace formula from \cite[Theorem 5.3]{bolte_trace_2009}. Let $\mc{W}_{n}\equiv\mc{W}_{n}(G)$ denote the set of closed walks with $n$ edges in the graph $G$. We represent a closed walk $\gamma\in\mc{W}_{n}$ by its sequence of vertices $\gamma=(i_{1},i_{2},...,i_{n})$, where $i_{n}$ is the last vertex visited before returning to $i_{1}$. The metric length of a closed walks $l_{\gamma}$ is the sum of the lengths of its constitutive edges:
\begin{align}
    l_{\gamma}&=l_{i_{1},i_{2}}+l_{i_{2},i_{3}}+\cdots+l_{i_{n},i_{1}}.
\end{align}
The amplitude $A_{\gamma}$ of a closed walk is the product of the entries of $U^{(i)}$ along the path:
\begin{align}
    A_{\gamma}&=U^{(i_{1})}_{f_{i_{1}}(i_{n}),f_{i_{1}}(i_{2})}U^{(i_{2})}_{f_{i_{2}}(i_{1}),f_{i_{2}}(i_{3})}\cdots U^{(i_{n})}_{f_{i_{n}}(i_{n-1}),f_{i_{n}}(i_{1})},
\end{align}
where $f_{i}$ is the function assigning an index in $[d_{i}]$ to each neighbour of $i$.

Since the matrix-valued function $U(\lambda)$ is: i) entire; ii) unitary for $\lambda\in\mbb{R}$; and iii) strictly sub-unitary for $\lambda\in\mbb{C}\setminus\mbb{R}$, we can use the argument principle as in \cite{bolte_trace_2009} to obtain the trace formula.
\begin{lemma}[Theorem 5.3 in \cite{bolte_trace_2009}]
Let $h\in C^{\infty}(\mbb{R})$ have a compactly supported Fourier transform. Let $\lambda_{j}=k_{j}^{2}$ denote the eigenvalues of a quantum graph. Then
\begin{align}
    \sum_{j=0}^{\infty}h(k_{j})&=\hat{h}(0)L+2\sum_{n=1}^{\infty}\sum_{\gamma\in\mc{W}_{n}}l_{i_{1},i_{2}}\Re(e^{i\lambda l_{\gamma}}A_{\gamma}\hat{h}(l_{\gamma})),\label{eq:traceFormula}
\end{align}
where $L=\|\mbf{l}\|_{1}$ is the sum of the lengths of all edges in $G$.
\end{lemma}

\section{Proof of Theorem \ref{thm1}}\label{s:Proof1}
The proof of Theorem \ref{thm1} follows the same outline as in the case of Weil--Petersson surfaces \cite{rudnick_goe_2023}. By the trace formula, the variance is given by the expectation of a double sum over closed walks on the random regular graph $G_{1}$. Since $h$ has compactly supported Fourier transform, the sums are truncated at walks of length $\eta^{-1}$. If $\eta^{-1}\ll N^{1/2}$, then we expect that most walks do not self-intersect and so the leading contribution should come from cycles. Using the known asymptotics for the expected number of cycles in the random regular graph, we can easily show that the restriction of the double sum to pairs of cycles differing only in starting vertex and orientation recovers the GOE variance. Thus we need to show that all other contributions to the double sum are negligible. We do this by instead bounding these contributions for the configuration model, which by a simple, well-known argument gives bounds for the random regular graph. In the configuration model, we can count pairs of closed walks by sampling two random walks one after the other and building the configuration simultaneously with the random walks, an idea which dates at least to Broder--Shamir \cite{broder_second_1987}. 

We begin by recalling the configuration model (see e.g. \cite{bollobas_probabilistic_1980} and in particular \cite{wormald_models_1999} for a review). We assign each vertex a set of $d$ half-edges and call a pairing of the $Nd$ half-edges a \textbf{configuration}. The set of configurations is denoted $\mc{P}_{N,d}$ and the configuration model corresponds to the uniform measure on $\mc{P}_{N,d}$. In general, a configuration corresponds to a multi-graph, since two half-edges belonging to the same vertex might be paired to form a loop, or several half-edges between the same two vertices paired to form multiple edges. However, the asymptotic probability as $N\to\infty$ that a pairing is \textit{simple}, i.e. corresponds to a graph, is well-known \cite{bender_asymptotic_1978, bollobas_probabilistic_1980}:
\begin{align}
    \mbb{P}(\textup{Simple})&\sim e^{\frac{1-d^{2}}{4}}.
\end{align}
In particular, it is bounded below for finite $d$. Let $\mc{E}$ be a subset of $\mc{G}_{N,d}$ and $\mc{E}'$ the corresponding subset of $\mc{P}_{N,d}$ (i.e. each pairing in $\mc{E}'$ corresponds to a graph in $\mc{E}$). Since each graph corresponds to the same number of simple pairings, we have
\begin{align}
    \mbb{P}(G\in\mc{E})&=\frac{\mbb{P}(P\in\mc{E}')}{\mbb{P}(\textup{Simple})}=C_{d}\mbb{P}(P\in\mc{E}').
\end{align}

We will apply this as follows. The set of closed walks $\mc{W}_{n}$ lifts straightforwardly to $\mc{P}_{N,d}$ by considering each configuration as a multi-graph. Let $\mc{W}(G)$ be a subset of the set $\mc{W}_{n_{1}}(G)\times\mc{W}_{n_{2}}(G)$ of pairs of closed walks of lengths $n_{1}$ and $n_{2}$ on $G$, and $\mc{W}(P)$ the corresponding subset of $\mc{W}_{n_{1}}(P)\times\mc{W}_{n_{2}}(P)$. Let
\begin{align}
    \mc{E}_{m}&:=\{G:|\mc{W}(G)|=m\},
\end{align}
be the subset of graphs with $m$ walks in $\mc{W}(G)$, $\mc{E}'_{m}$ the subset of configurations corresponding to graphs in $\mc{E}_{m}$, and
\begin{align}
    \mc{E}''_{m}&:=\{P:|\mc{W}(P)|=m\}
\end{align}
the subset of configurations with $m$ walks in $\mc{W}(P)$. Then we have $\mc{E}'_{m}\subset\mc{E}''_{m}$ and
\begin{align}
    \mbb{E}|\mc{W}(G)|&=\sum_{m\geq0}m\mbb{P}(\mc{E}_{m})\nonumber\\
    &=C_{d}\sum_{m\geq0}m\mbb{P}(\mc{E}'_{m})\nonumber\\
    &\leq C_{d}\sum_{m\geq0}m\mbb{P}(\mc{E}''_{m})\nonumber\\
    &=C_{d}\mbb{E}|\mc{W}(P)|.
\end{align}
In other words, to bound the expected cardinality of a given subset of closed walks in the random regular graph, it is enough to bound the corresponding expectation for the configuration model. We record this observation in the following lemma.
\begin{lemma}\label{lem:configuration}
Let $\mc{W}$ be a subset of closed walks. Then
\begin{align}
    \mbb{E}|\mc{W}(G)|&\leq C_{d}\mbb{E}|\mc{W}(P)|.
\end{align}
\end{lemma}

The next three lemmas bound certain subsets of closed geodesics. Before stating them, we define some terminology.
\begin{itemize}
\item A closed walk is \textbf{simple} if it visits all vertices but the starting vertex exactly once.
\item A closed walk is \textbf{primitive} if it is not a repetition of a shorter walk, otherwise it is \textbf{composite}.
\item A \textbf{backtrack} is a step that revisits the preceeding vertex.
\item A \textbf{geodesic} is a non-backtracking walk.
\end{itemize} 
We have the following well-known asymptotic formula for the expected number of cycles.
\begin{lemma}[Eq. (2.5) and Lemma 4 in \cite{garmo_asymptotic_1999}]\label{lem:cycles}
Let $\mc{C}_{m}(G)$ denote the number of cycles in $G$. Then for any $m=o(N)$ we have
\begin{align}
    \mbb{E}\mc{C}_{m}(G)&=\left[1+O\left(\frac{m}{N}\right)\right]\frac{(d-1)^{m}}{2m}.
\end{align}
\end{lemma}
Note that for each cycle of length $m$ there are $2m$ closed geodesics, corresponding to $m$ choices of start vertex and 2 choices of direction.

Next we bound the expected number of non-simple closed geodesics, which are either composite, non-simple, or both.
\begin{lemma}\label{lem:non-cycles}
Let $\mc{W}_{m,NS}\subset\mc{W}_{m}$ denote the subset of non-simple closed geodesics of length $m$. Then
\begin{align}
    \mbb{E}|\mc{W}_{m,NS}|&\lesssim\frac{m^{2}}{N}\cdot (d-1)^{m}+\sum_{k|m}(d-1)^{m/k}.\label{eq:LNS}
\end{align}
\end{lemma}
\begin{proof}
This is essentially Claims 3.2 and 3.3 from \cite{dozier_simple_2024}. The second term in \eqref{eq:LNS} comes from composite geodesics that are repetitions of shorter simple geodesics. For the first term, we use Lemma \ref{lem:configuration} and bound the expected number of primitive, non-simple closed geodesics in the configuration model as in \cite[Section 3.1]{dozier_simple_2024}.
\end{proof}

Next we bound the expected number of pairs of closed geodesics that share a given number of steps.
\begin{lemma}\label{lem:pairs}
Let
\begin{align}
    \mc{W}(m_{1},m_{2},s,p):=\bigl\{(\gamma_{1},\gamma_{2})\in\mc{W}_{m_{1}}\times\mc{W}_{m_{2}}:&\textup{ $\gamma_{1}$ and $\gamma_{2}$ have $s$ steps in common}\nonumber\\
    &\textup{that occur in $p$ segments}\bigr\}.
\end{align}
Then
\begin{align}
    \mbb{E}|\mc{W}(m_{1},m_{2},s,p)|&\lesssim2^{p}p!\begin{pmatrix}p+s-1\\p-1\end{pmatrix}\begin{pmatrix}m_{1}\\p\end{pmatrix}\begin{pmatrix}m_{2}\\p\end{pmatrix}\left(\frac{d(m_{1}\vee m_{2})}{N}\right)^{p}(d-1)^{m_{1}+m_{2}-s}.
\end{align}\label{eq:L(m_{1},m_{2},s,p)}
\end{lemma}
\begin{proof}
We will sample uniformly from all pairs of non-backtracking walks and bound the probability that the resulting pair is in $\mc{W}(m_{1},m_{2},s,p)$. Furthermore, we sample the pair of walks by performing two random walks one after the other and building the configuration simultaneously. Let $\gamma_{1}$ and $\gamma_{2}$ be the first and second random walks. Requiring $\gamma_{1}$ and $\gamma_{2}$ to be closed forces the last vertices to have a half-edge paired with the first, which has probability at most
\begin{align}
	\frac{Cm_{1}m_{2}}{N^{2}}.
\end{align}
For the second walk $\gamma_{2}$, the first step of each segment of common edges forces a pairing with a half-edge of a vertex of $\gamma_{1}$, which occurs with probability at most
\begin{align}
	\frac{Cm_{1}}{N}.
\end{align}
For the remaining steps in each segment, the pairing is already chosen but the choice of half-edge along which to take the next step is fixed to be the same as in $\gamma_{1}$. Thus each such step reduces the probability by a factor of
\begin{align}
	\frac{1}{d-1}.
\end{align}
Taking the product over each step and each segment, we obtain a factor
\begin{align}
	\Biggl(\frac{Cm_{1}(d-1)}{N}\Biggr)^{p}\cdot\frac{1}{(d-1)^{s}}.
\end{align}
There are at most
\begin{align}
	\begin{pmatrix}m_{1}\\p\end{pmatrix}\begin{pmatrix}m_{2}\\p\end{pmatrix}
\end{align}
ways to choose the common segments in $\gamma_{1}$ and $\gamma_{2}$, 
\begin{align}
	\begin{pmatrix}p+s-1\\p-1\end{pmatrix}
\end{align}
ways to distribute $s$ edges among $p$ segments, and $2^{p}p!$ ways to choose the orientation and order of the segments in $\gamma_{2}$ with respect to $\gamma_{1}$. Altogether, the probability that $(\gamma_{1},\gamma_{2})\in\mc{W}(m_{1},m_{2},s,p)$ is bounded above by
\begin{align}
	2^{p}p!\begin{pmatrix}p+s-1\\p-1\end{pmatrix}\begin{pmatrix}m_{1}\\p\end{pmatrix}\begin{pmatrix}m_{2}\\p\end{pmatrix}\cdot\Biggl(\frac{Cd(m_{1}\vee m_{2})}{N}\Biggr)^{p}\cdot\frac{m_{1}m_{2}}{N^{2}}\cdot\frac{1}{(d-1)^{s}}.
\end{align}
After multiplying by the total number $N^{2}d^{2}(d-1)^{m_{1}+m_{2}}$ of pairs of non-backtracking walks of lengths $m_{1}$ and $m_{2}$, we obtain \eqref{eq:L(m_{1},m_{2},s,p)}.
\end{proof}

Using Lemma \ref{lem:cycles}, we can easily bound the expectation.
\begin{lemma}\label{lem:expectation}
There is a $c_{\mu}>0$ such that, for $\lambda>c_{\mu}\sqrt{d-1}$, we have
\begin{align}
	\mbb{E}\mc{L}_{\lambda,\eta,h}(\mc{G})&=\hat{h}(0)\mbb{E}L+O(\eta^{2}).
\end{align}
\end{lemma}
\begin{proof}
By the trace formula we have
\begin{align}
	\mbb{E}\mc{L}_{\lambda,\eta,h}(\mc{G})-\hat{h}(0)\mbb{E}L&=2\Re\sum_{n=1}^{\infty}\mbb{E}\sum_{\gamma\in\mc{W}_{n}}l_{e_{1}(\gamma)}\hat{h}(l_{\gamma}\eta)e^{i\lambda l_{\gamma}}A_{\gamma}.
\end{align}
Taking the expectation of the lengths first we obtain
\begin{align}
	\mbb{E}_{\mbf{l}}l_{e_{1}(\gamma)}e^{i\lambda l_{\gamma}}A_{\gamma}\hat{h}(l_{\gamma}\eta)&=O\left(\frac{c_{\mu}^{n}}{(d-1)^{n/2}\lambda^{n}}\right),
\end{align}
for some constant $c_{\mu}>0$, which follows from integration by parts and the fact that $A_{\gamma}=O((d-1)^{-n/2})$. Using Lemmas \ref{lem:cycles} and \ref{lem:non-cycles} to bound $\mbb{E}|\mc{W}_{\gamma}|$ we find
\begin{align}
	\Bigl|\mbb{E}\mc{L}_{\lambda,\eta,h}(\mc{G})-\hat{h}(0)\mbb{E}L\Bigr|&\lesssim\eta^{2}\sum_{n\geq1}\Biggl(\frac{c^{2}_{\mu}(d-1)}{\lambda^{2}}\Biggr)^{n}\\
	&=O(\eta^{2}).
\end{align} 
\end{proof}

Let $\mc{W}^{(d)}_{n}$ denote the set of pairs $(\gamma_{1},\gamma_{2})$ of closed geodesics such that $\gamma_{2}$ traverses the same set of edges as $\gamma_{1}$ with the same multiplicities. Let 
\begin{align}
	\mc{W}^{(o)}(m_{1},m_{2})&=\begin{cases}
	(\mc{W}_{m_{1}}\times\mc{W}_{m_{1}})\setminus\mc{W}^{(d)}_{n}&\quad m_{1}=m_{2}\\
	\mc{W}_{m_{1}}\times\mc{W}_{m_{2}}&\quad m_{1}\neq m_{2}
	\end{cases}.
\end{align}
From the above lemma we can write
\begin{align}
	\mbb{E}\Bigl|\mc{L}_{\lambda,\eta,h}-\mbb{E}\mc{L}_{\lambda,\eta,h}\Bigr|^{2}&=\mc{V}_{\lambda,\eta}^{(d)}(h)+\mc{V}_{\lambda,\eta}^{(o)}(h)+O(\eta^{2}),
\end{align}
where
\begin{align}
	\mc{V}_{\lambda,\eta}^{(d)}(h)&=2\eta^{2}\sum_{n=1}^{\infty}\sum_{(\gamma_{1},\gamma_{2})\in\mc{W}^{(d)}_{n}}l^{2}_{e_{1}}|\hat{h}(l_{\gamma_{1}}\eta)|^{2}A_{\gamma_{1}}\bar{A}_{\gamma_{2}},
\end{align}
and
\begin{align}
	\mc{V}_{\lambda,\eta}^{(o)}(h)&=2\eta^{2}\sum_{n_{1},n_{2}=1}^{\infty}\sum_{(\gamma_{1},\gamma_{2})\in\mc{W}^{(o)}(n_{1},n_{2})}l_{e_{1}}l_{e'_{1}}\hat{h}(l_{\gamma_{1}}\eta)\bar{\hat{h}}(l_{\gamma_{2}}\eta)A_{\gamma_{1}}\bar{A}_{\gamma_{2}}.
\end{align}
We have split the variance into the diagonal and off-diagonal terms, where the former contain the leading contributions.

\begin{lemma}\label{lem:diagonal}
We have
\begin{align}
	\mc{V}_{\lambda,\eta}^{(d)}(h)&=2\int_{-\infty}^{\infty}x|\hat{h}(x)|^{2}\diff x+O\Biggl(\eta^{1/2}+\frac{1}{N\eta^{2}}\Biggr).\label{eq:diagonal}
\end{align}
\end{lemma}
\begin{proof}
We split the diagonal term into pairs of simple (s) and non-simple (ns) closed geodesics:
\begin{align}
	\mc{V}_{\lambda,\eta}^{(d)}(h)&=\mc{V}_{\lambda,\eta}^{(d,s)}(h)+\mc{V}_{\lambda,\eta}^{(d,ns)}(h).
\end{align}
If $\gamma_{1}$ is simple, then $\gamma_{2}$ must be a cyclic permutation and/or reversal of $\gamma_{1}$. For a given cycle, there are $4n^{2}$ such pairs of closed paths: $n$ choices of starting point for each cycle and 2 choices of relative orientation. The contribution from these pairs is
\begin{align}
    \mc{V}_{\lambda,\eta}^{(d,s)}(h)&=8\eta^{2}\sum_{n=1}^{\infty}\frac{n^{2}}{(d-1)^{n}}\mbb{E}\sum_{\gamma\in \mc{C}_{n}}l_{e_{1}(\gamma_{1})}l_{e_{1}(\gamma_{2})}\bigl|\hat{h}(l_{\gamma}\eta)\bigr|^{2}.\label{eq:cycle}
\end{align}
Here we used the fact that $A_{\gamma_{1}}\bar{A}_{\gamma_{2}}=(d-1)^{-n}$ for each such pair since the vertex unitaries are symmetric and equi-transmitting. Since the lengths are iid and each appears exactly once in a cycle, the expectation over the lengths depends only on the topological length of the path:
\begin{align}
    \mbb{E}_{l}l_{e_{1}(\gamma_{1})}l_{e_{1}(\gamma_{2})}\left|\hat{h}(l_{\gamma}\eta)\right|^{2}&=f(n).
\end{align}
For a fixed $n$, this quantity depends on the distribution of lengths and is therefore not universal. In the mesoscopic regime, however, the contribution from lengths $n\leq \eta^{-1+\epsilon}$ is $O(\eta^{2\epsilon})$ for any $\epsilon>0$. The only surviving terms come from paths with lengths growing with $\eta^{-1}$. For these we can make use of Hoeffding's inequality:
\begin{align}
    \mbb{P}\left(|l_{\gamma}-nl|\geq t\right)&\leq2\exp\left(-\frac{2t^{2}}{n(l_{max}-l_{min})^{2}}\right),
\end{align}
where $l$ is the mean length and $l_{min},\,l_{max}$ the edges of the support. Thus we have
\begin{align}
    l_{\gamma}&=nl+O\left(n^{1/2+\epsilon}\right),
\end{align}
with probability at least $1-e^{-cn^{2\epsilon}}$ and hence
\begin{align}
    f(n)&=\bigl[1+O(\eta^{1/2})\bigr]l^{2}\bigl|\hat{h}(nl\eta)\bigr|^{2},
\end{align}
where $l=\mbb{E}l_{e}$ is the mean length. Inserting this into \eqref{eq:cycle} we obtain
\begin{align}
    \bigl[1+O(\eta^{1/2})\bigr]\times8\eta^{2}\sum_{n=1}^{\infty}\frac{n^{2}l^{2}}{(d-1)^{n}}\bigl|\hat{h}(nl\eta)\bigr|^{2}\mbb{E}|C_{n}|.
\end{align}
By Lemma \ref{lem:cycles} we have
\begin{align}
    \mbb{E}|\mc{C}_{n}|&=\left[1+O\left(\frac{n(d+n)}{N}\right)\right]\frac{(d-1)^{n}}{2n}.
\end{align}
Thus we find
\begin{align}
    \mc{V}_{\lambda,\eta}^{(d,s)}(h)&=4l^{2}\eta\sum_{n=1}^{\infty}n\eta\bigl|\hat{h}(nl\eta)\bigr|^{2}+O\left(\eta^{1/2}+\frac{n(d+n)}{N}\right).
\end{align}
Approximating the sum by a Riemann integral we obtain
\begin{align}
    \mc{V}_{\lambda,\eta}^{(d,s)}(h)&=4\int_{0}^{\infty}x|\hat{h}(x)|^{2}\diff x+O\left(\eta^{1/2}+\frac{n(d+n)}{N}\right).
\end{align}

By Lemma \ref{lem:non-cycles}, we can bound $\mc{V}^{(d,ns)}$ by
\begin{align}
	|\mc{V}^{(d,ns)}_{\lambda,\eta}(h)|&\leq C\eta\sum_{n\geq1}n\eta\bigl|\hat{h}(nl\eta)\bigr|^{2}\Biggl(\frac{n^{2}}{N}+\frac{1}{(d-1)^{n/2}}\Biggr)\leq\frac{C}{N\eta^{2}}.
\end{align}
Combined with the estimate for $\mc{V}^{(d,s)}_{\lambda,\eta}(h)$, we obtain \eqref{eq:diagonal}.
\end{proof}

\begin{lemma}\label{lem:off-diagonal}
We have
\begin{align}
	|\mc{V}^{(o)}_{\lambda,\eta}(h)|&\leq\frac{C}{N\eta^{3}}.
\end{align}
\end{lemma}
\begin{proof}
The set of off-diagonal pairs can be decomposed into $\mc{W}(n_{1},n_{2},s,p)$ for $0\leq s\leq n_{1}\wedge n_{2}$ and $1\wedge s \leq p\leq s$. Consider a term in $\mc{V}^{(o)}_{\lambda,\eta}(h)$ corresponding to a pair in $\mc{W}(n_{1},n_{2},s,p)$. The dependence of such a term on $\mbf{l}$ is
\begin{align}
	l_{e_{1}(\gamma_{1})}l_{e_{1}(\gamma_{2})}e^{i\lambda(l_{\gamma_{1}}-l_{\gamma_{2}})}\hat{h}(l_{\gamma_{1}}\eta)\bar{\hat{h}}(l_{\gamma_{2}}\eta).
\end{align}
Since $l_{\gamma_{1}}-l_{\gamma_{2}}$ is a sum of $n_{1}+n_{2}-2s$ lengths, the expectation over $\mbf{l}$ contributes a factor of $(c_{\mu}/\lambda)^{n_{1}+n_{2}-2s}$, leading to the bound
\begin{align}
	|\mc{V}^{(o)}_{\lambda,\eta}(h)|&\leq C\eta^{2}\sum_{n_{1},n_{2}\leq\eta^{-1}}\sum_{s=0}^{n_{1}\wedge n_{2}}\sum_{p=1\wedge s}^{s}\frac{c_{\mu}^{n_{1}+n_{2}-2s}\mbb{E}|\mc{W}(n_{1},n_{2},s,p)|}{(d-1)^{(n_{1}+n_{2})/2}\lambda^{n_{1}+n_{2}-2s}}.
\end{align}
Inserting the bound from Lemma \ref{lem:pairs} we find
\begin{align}
	|\mc{V}^{(o)}_{\lambda,\eta}(h)|&\leq C\eta^{2}\sum_{n_{1},n_{2}\leq\eta^{-1}}\sum_{s=0}^{n_{1}\wedge n_{2}}\sum_{p=1\wedge s}^{s}\Biggl(\frac{(n_{1}\vee n_{2})^{3}}{N}\Biggr)^{p}\cdot\Biggl(\frac{c_{\mu}(d-1)}{\lambda^{2}}\Biggr)^{(n_{1}+n_{2})/2-s}\\
	&\leq\frac{C}{N\eta^{3}}.
\end{align}

\end{proof}

Theorem \ref{thm1} follows immediately from Lemmas \ref{lem:diagonal} and \ref{lem:off-diagonal}.

\section{Proof of Theorem \ref{thm2}}\label{s:Proof2}
Before we begin, we observe that the choice of ordering $f_{v}$ is immaterial and we will therefore use the abuse of notation $f_{v}(w)=w$, i.e. whenever $w$ appears as an index of a unitary $U^{(v)}$ it is understood to mean the ordered version $f_{v}(w)$.

Since the expectation of a polynomial in $U^{(i)}$ and $U^{(i)}$ is nonzero only if the degrees of $U^{(i)}$ and $U^{(i)}$ are equal, we immediately obtain
\begin{align}
    \mbb{E}\mc{L}_{\lambda,\eta,h}(\mc{G})&=\hat{h}(0)L,
\end{align}
and
\begin{align}
    \mc{V}_{\lambda,\eta}(h)&=2\eta^{2}\sum_{n=1}^{\infty}\sum_{(\gamma_{1},\gamma_{2})\in\mc{W}_{n}\times\mc{W}_{n}}l_{i_{1},i_{2}}l_{j_{1},j_{2}}e^{i\lambda(l_{\gamma_{1}}-l_{\gamma_{2}})}\hat{h}\left(l_{\gamma_{1}}\eta\right)\bar{\hat{h}}\left(l_{\gamma_{2}}\eta\right)\mbb{E}A_{\gamma_{1}}\bar{A}_{\gamma_{2}}.
\end{align}

We can further simplify the double sum by noting that, for a Haar unitary $V$, 
\begin{align}
    \mbb{E}V_{x_{1},y_{1}}\cdots V_{x_{m},y_{m}}\bar{V}_{u_{1},v_{1}}\cdots\bar{V}_{u_{m},v_{m}}
\end{align}
is nonzero only if $(u_{1},...,u_{m})$ is a permutation of $(x_{1},...,x_{m})$ and $(v_{1},...,v_{m})$ is a permutation of $(y_{1},...,y_{m})$. This means that $\gamma_{1}$ and $\gamma_{2}$ must traverse the same set of edges with the same multiplicities (but possibly in a different order). In particular, we have $l_{\gamma_{1}}=l_{\gamma_{2}}$. If we define $\mc{W}(\gamma)$ to be the set of closed paths $\gamma'$ for which $\mbb{E}A_{\gamma}\bar{A}_{\gamma'}$ is nonzero, we can write
\begin{align}
    \mc{V}_{\lambda,\eta}(h)&=2\eta^{2}\sum_{n=1}^{\infty}\sum_{\gamma_{1}\in\mc{W}_{n}}l_{i_{1},i_{2}}\left|\hat{h}\left(l_{\gamma_{1}}\eta\right)\right|^{2}\sum_{\gamma_{2}\in\mc{W}(\gamma_{1})}l_{j_{1},j_{2}}\mbb{E}A_{\gamma_{1}}\bar{A}_{\gamma_{2}}.\label{eq:simplified}
\end{align}

We will group closed paths $\gamma_{1}$ according to the number $v$ of distinct indices and their multiplicities $m_{j},\,j=1,...,v$, where the return to the starting vertex is \textit{not} counted. In general we have
\begin{align}
    \sum_{j=1}^{v}m_{j}=n,\qquad m_{j}\geq1,
\end{align}
i.e. $\mbf{m}=(m_{1},...,m_{v})$ is a partition of $n$ into $v$ parts. The closed path is a cycle if and only if $v=n$ and $m_{j}=1,\,j=1,...,v$. Let $\mc{N}(n,v,\mbf{m})$ denote the set of closed paths in $K_{N+1}$ with prescribed $(n,v,\mbf{m})$. Once $v$ vertices have been chosen from $N+1$, we enumerate them in order of appearance in the closed path. This gives a word $w=w_{1}w_{2}\cdots w_{n}$ of length $n$ with letters in $[v]$ such that
\begin{align}
    w_{j}&\leq\max_{i<j}w_{i}+1,\\
    w_{j+1}&\neq w_{j}.\label{eq:cond2}
\end{align}
The first condition is precisely the definition of restricted growth functions (RGF), while the second arises because we do not have self loops in $K_{N+1}$. Thus we have
\begin{align}
    |\mc{N}(n,v,\mbf{m})|&=(N+1)_{v}|\mc{R}(n,v,\mbf{m})|,
\end{align}
where $(N+1)_{v}$ is the falling factorial and $\mc{R}(n,v,\mbf{m})$ is the set of RGFs satisfying \eqref{eq:cond2} and with letter multiplicities $\mbf{m}$.

It is well known that RGFs are in bijection with set partitions of $[n]$, which can be seen as follows. For a given partition, order the blocks $B_{1},...,B_{v}$ such that $1\in B_{1}$ and the first element of $B_{j}$ is the smallest integer that has not appeared in $B_{i},\,i<j$. Construct an RGF $w=w_{1}\cdots w_{n}$ by setting $w_{j}$ equal to the index of the block to which $j$ belongs. Conversely, given $w=w_{1}\cdots w_{n}$, group the indices of equal letters $w_{j}$ into blocks to construct a partition. In the context of set partitions, the condition \eqref{eq:cond2} requires that no block contains consecutive integers.

Clearly, $|\mc{R}(n,v,\mbf{m})|$ is upper bounded by the number obtained after removing the condition \eqref{eq:cond2}, which is simply the number of partitions of $[n]$ into blocks of sizes $m_{1},...,m_{k}$. If we define
\begin{align}
    r_{k}&:=|\{j:m_{j}=k\}|,
\end{align}
then we have
\begin{align}
    |\mc{N}(n,v,\mbf{m})|&\leq\frac{(N+1)_{v}}{\prod_{k\geq1} r_{k}!}\begin{pmatrix}n\\m_{1},...,m_{k}\end{pmatrix}.\label{eq:NBound}
\end{align}
We remark that if we do not specify $\mbf{m}$ and merely count the number of RGFs satisfying \eqref{eq:cond2}, then this is exactly $S(n-1,v-1)$ where $S(m,k)$ is the Stirling number of the second kind \cite[Eq. (3)]{munagi_combinatorial_2016}.

We now turn to the factor $\mbb{E}A_{\gamma_{1}}\bar{A}_{\gamma_{2}}$:
\begin{align}
    \mbb{E}A_{\gamma_{1}}\bar{A}_{\gamma_{2}}&=\mbb{E}U^{(i_{1})}_{i_{n},i_{2}}\cdots U^{(i_{n})}_{i_{n-1},i_{1}}\bar{U}^{(j_{1})}_{j_{n},j_{2}}\cdots\bar{U}^{(j_{n})}_{j_{n-1},j_{1}}.
\end{align}
As discussed above, $\gamma_{2}$ must correspond to the same multiset of edges as $\gamma_{1}$, and in particular the same number of distinct vertices $v$ and multiplicities $\mbf{m}$. Since the $U^{(i)}$ are mutually independent, the expectation factorises into a product of $v$ separate expectations of the form
\begin{align}
    \mbb{E}U^{(i)}_{x_{1},y_{1}}\cdots U^{(i)}_{x_{m_{j}},y_{m_{j}}}\bar{U}^{(i)}_{u_{1},v_{1}}\cdots\bar{U}^{(i)}_{u_{m_{j}},v_{m_{j}}}.
\end{align}
By H\"{o}lder's inequality, any such expectation is bounded above by
\begin{align}
    \mbb{E}|U^{(i)}_{x_{1},y_{1}}|^{2m_{j}}&=\frac{N^{m_{j}}(N-1)!}{(N+m_{j}-1)!}\cdot\frac{m_{j}!}{N^{m_{j}}}\leq\frac{m_{j}!}{N^{m_{j}}}.
\end{align}
Using this crude bound we obtain
\begin{align}
    |\mbb{E}A_{\gamma_{1}}\bar{A}_{\gamma_{2}}|&\leq\prod_{j=1}^{v}\frac{m_{j}!}{N^{m_{j}}}=\frac{\prod_{j=1}^{v}m_{j}!}{N^{n}}.\label{eq:ABound}
\end{align}

Let us assume for the moment that all lengths are equal to $l$. Then \eqref{eq:simplified} becomes
\begin{align}
    \mc{V}_{\lambda,\eta}(h)&=2l^{2}\eta^{2}\sum_{n=1}^{\infty}\left|\hat{h}\left(nl\eta\right)\right|^{2}\sum_{v=2}^{n}\sum_{\mbf{m}\in\mc{P}(n,v)}\sum_{\gamma_{1}\in\mc{N}(n,v,\mbf{m})}\sum_{\gamma_{2}\in\mc{W}(\gamma_{1})}\mbb{E}A_{\gamma_{1}}\bar{A}_{\gamma_{2}},
\end{align}
where $\mc{P}(n,v)$ is the set of integer partitions of $n$ into $v$ parts. We call the contribution from the summand with $v=n$ the \textit{dominant} contribution and the rest the \textit{sub-leading} contribution, and denote these by $\mc{V}^{(d)}_{\lambda,\eta}(h)$ and $\mc{V}^{(s)}_{\lambda,\eta}(h)$ respectively. Using the bound \eqref{eq:ABound} we obtain
\begin{align}
    |\mc{V}^{(s)}_{\lambda,\eta}(h)|&\leq2l^{2}\eta^{2}\sum_{n=1}^{\infty}\left|\hat{h}\left(nl\eta\right)\right|^{2}\sum_{v=2}^{n-1}\sum_{\mbf{m}\in\mc{P}(n,v)}\sum_{\gamma_{1}\in\mc{N}(n,v,\mbf{m})}\frac{|\mc{W}(\gamma_{1})|}{N^{n}}\prod_{j=1}^{v}m_{j}!.\label{eq:Vo1}
\end{align}
We bound $|\mc{W}(\gamma_{1})|$ as follows. A closed path $\gamma_{1}$ with $v$ vertices and vertex multiplicities $\mbf{m}$ naturally induces a directed graph, which we turn into a directed multigraph by splitting each edge into multiple edges according to the number of times it is traversed by the closed path. The in- and out-degree of vertex $j$ in the multigraph are both equal to $m_{j}$. The original closed path is then identified with an Eulerian cycle on this multigraph, i.e. a cycle that traverses each edge exactly once. For $\gamma_{2}$ to belong to $\mc{W}(\gamma_{1})$, it must (at the very least) give rise to the same multigraph. Thus $|\mc{W}(\gamma_{1})|$ is bounded above by the number of Eulerian cyles on a directed multigraph with $v$ vertices having in- and out-degrees $m_{j},\,j=1,...,v$. The latter number is bounded above by
\begin{align}
    v\prod_{j=1}^{v}m_{j}!,\label{eq:Eulerian}
\end{align}
since we have $v$ choices of starting vertex and $m_{j}!$ ways to match the edges ending at $j$ to those beginning at $j$. 

Let us now impose the condition $\eta\geq N^{-1/2}$. Since the support of $\hat{h}$ is compact, this truncates the outermost sum to $n\leq C\sqrt{N}$. In this region, since $v\leq n$ we have $(N+1)_{v}\leq CN^{v}$ and so using the bounds in \eqref{eq:NBound}, \eqref{eq:ABound} and \eqref{eq:Eulerian}, we find
\begin{align}
    |\mc{V}^{(s)}_{\lambda,\eta}(h)|&\leq Cl\eta\sum_{n=1}^{\infty}nl\eta\left|\hat{h}\left(nl\eta\right)\right|^{2}\sum_{v=2}^{n-1}\sum_{\mbf{m}\in\mc{P}(n,v)}\frac{n!\prod_{j=1}^{v}m_{j}!}{N^{n-v}\prod_{k\geq1}r_{k}!}\label{eq:Vo2}
\end{align}
We rewrite \eqref{eq:Vo2} as follows
\begin{align}
    |\mc{V}^{(s)}_{\lambda,\eta}(h)|&\leq nl\eta\sum_{n=1}^{\infty}nl\eta\left|\hat{h}\left(nl\eta\right)\right|^{2}\left(\frac{n!}{N^{n}}S_{N,n}-1-\frac{(n!)^{2}}{N^{n-1}}\right),
\end{align}
where
\begin{align}
    S_{N,n}&:=\sum_{v=1}^{n}N^{v}\sum_{\mbf{m}\in\mc{P}(n,v)}\prod_{k\geq1}\frac{(k!)^{r_{k}}}{r_{k}!}.
\end{align}
We have subtracted $n!/N$ and $n!N$ from $S_{N,n}$ since these correspond to the summands with $v=n$ and $v=1$ respectively.

Observe that
\begin{align}
    \sum_{r_{1}+\cdots+nr_{n}=n}\prod_{k=1}^{n}f(k,r_{k})&=[z^{n}]\prod_{k=1}^{n}\sum_{r=0}^{\infty}f(k,r)z^{kr}.
\end{align}
In our case we have
\begin{align}
    S_{N,n}&=\sum_{r_{1}+\cdots+nr_{n}=n}\prod_{k=1}^{n}f(k,r_{k}),
\end{align}
with
\begin{align}
    f(k,r)&=\frac{(k!N)^{r}}{r!}.
\end{align}
Thus
\begin{align}
    S_{N,n}&=[z^{n}]\prod_{k=1}^{n}\sum_{r=0}^{\infty}\frac{(k!Nz^{k})^{r}}{r!}=[z^{n}]\prod_{k=1}^{n}e^{k!Nz^{k}}=[z^{n}]e^{NP_{n}(z)},
\end{align}
where
\begin{align}
    P_{n}(z)&:=\sum_{k=1}^{n}k!z^{k}.
\end{align}
By Cauchy's integral formula, for any $r>0$ we have
\begin{align}
    S_{N,n}&=\frac{1}{r^{n}}\int_{-\pi}^{\pi}e^{NP_{n}(re^{i\theta})-in\theta}\frac{\diff \theta}{2\pi}.
\end{align}
Let $r=r_{n}$ solve
\begin{align}
    rP'_{n}(r)&=\frac{n}{N},
\end{align}
or equivalently
\begin{align}
    \sum_{k=1}^{n}k!kr^{k}&=\frac{n}{N}.
\end{align}
Since the coefficients are positive, there is a unique solution in the interval $[0,n/N]$. Since 
\begin{align}
    \frac{(k+1)!(k+1)(n/N)^{k+1}}{k!k(n/N)^{k}}&=\frac{(k+1)^{2}n}{kN}=O\left(\frac{n^{2}}{N}\right),
\end{align}
if $n=o(N^{1/2})$ then, for any $l=1,...,n$, we have
\begin{align}
    \sum_{k=l+1}^{n}k!kr_{n}^{k}&\leq\sum_{k=l+1}^{n}k!k\left(\frac{n}{N}\right)^{k}\\
    &\leq\frac{Cl!ln^{l}}{N^{l}}\sum_{k=l+1}^{n}\left(\frac{n^{2}}{N}\right)^{k-l}\\
    &\leq\frac{Cn^{2}}{N}\cdot\frac{l!ln^{l}}{N^{l}},
\end{align}
and hence
\begin{align}
    \sum_{k=1}^{l}k!kr_{n}^{k}+o\left(\frac{l!ln^{l}}{N^{l}}\right)&=\frac{n}{N}.
\end{align}
This implies the asymptotic expansion
\begin{align}
    r_{n}&=\frac{n}{N}-\frac{4n^{2}}{N^{2}}+O\left(\frac{n^{3}}{N^{3}}\right).
\end{align}
By Laplace's method, we obtain
\begin{align}
    S_{N,n}&=\frac{e^{NP_{n}(r_{n})}}{r_{n}^{n}\sqrt{2\pi N(r_{n}P'_{n}(r_{n})+r^{2}_{n}P''_{n}(r_{n}))}}\left[1+O\left(\frac{1}{\sqrt{n}}\right)\right]\\
    &=\frac{1}{\sqrt{2\pi n}}\left(\frac{Ne}{n}\right)^{n}\left[1+O\left(\frac{n^{2}}{N}+\frac{1}{\sqrt{n}}\right)\right].
\end{align}    
Thus, by Stirling's formula,
\begin{align}
    \frac{n!}{N^{n}}S_{N,n}&=1+O\left(\frac{n^{2}}{N}+\frac{1}{\sqrt{n}}\right).
\end{align}
Since moreover
\begin{align}
    \frac{(n!)^{2}}{N^{n-1}}&\leq\frac{2\pi e^{-2n+2}n^{2n+1}}{N^{n-1}}=O\left(\frac{n^{2}}{N}\right),
\end{align}
for $n=o(N^{1/2})$, we have obtained
\begin{align}
    |\mc{V}^{(s)}_{\lambda,\eta}(h)|&\leq Cl\eta\sum_{n=1}^{\infty}nl\eta\left|\hat{h}\left(nl\eta\right)\right|^{2}\times O\left(\frac{n^{2}}{N}+\frac{1}{\sqrt{n}}\right)\nonumber\\
    &=O\left(\frac{1}{N\eta^{2}}+\|\hat{h}'\|_{\infty}\sqrt{\eta}\right).
\end{align}
Note that if $n\gg N^{1/2}$, then one can show that
\begin{align}
    \frac{n!}{N}S_{N,n}-1-\frac{(n!)^{2}}{N^{n-1}}
\end{align}
is unbounded as $n\to\infty$. In brief, we deduce that in this case
\begin{align}
    r_{n}>\frac{1-\epsilon}{(n!N)^{1/n}}
\end{align}
for some $\epsilon\in(0,1)$, and hence $NP_{n}(r_{n})>Nr_{n}>(1-\delta)N/n$.

The dominant contribution corresponds to $v=n$, i.e. $\gamma_{1}$ is a cycle of length $n$, of which there are $(N+1)_{n}$ possible choices in $K_{N+1}$. In this case $\mc{W}(\gamma_{1})$ consists of the $n$ cyclic permutations of $\gamma_{1}$ and we have
\begin{align}
    \mbb{E}A_{\gamma_{1}}\bar{A}_{\gamma_{2}}&=\mbb{E}|U^{(i_{1})}_{i_{n},i_{2}}|^{2}\cdots|U^{(i_{n})}_{i_{n-1},i_{1}}|^{2}=\frac{1}{N^{n}}
\end{align}
for all $\gamma_{2}\in\mc{W}(\gamma_{1})$. Hence
\begin{align}
    \mc{V}^{(d)}_{\lambda,\eta}(h)&=2l\eta\sum_{n=1}^{\infty}nl\eta\left|\hat{h}\left(nl\eta\right)\right|^{2}\cdot\frac{(N+1)_{n}}{N^{n}}\\
    &=\left[1+O\left(\frac{n^{2}}{N}\right)\right]2l\eta\sum_{n=1}^{\infty}nl\eta\left|\hat{h}\left(nl\eta\right)\right|^{2}\\
    &=\left[1+O\left(\frac{n^{2}}{N}+\|\hat{h}'\|_{\infty}\sqrt{\eta}\right)\right]2\int_{0}^{\infty}x|\hat{h}(x)|^{2}\diff x,
\end{align}
where in the last line we approximated the sum by a Riemann integral.

We can repeat the same arguments in the case where the lengths are not all equal as long as they are uniformly bounded. Now the dominant contribution takes the form
\begin{align}
    \mc{V}^{(d)}_{\lambda,\eta}(h)&=\left[1+O\left(\frac{n}{N}\right)\right]2\eta\sum_{n=1}^{\infty}n\eta\cdot\frac{1}{N^{n}}\sum_{\gamma_{1}}l_{i_{1},i_{2}}l_{i_{2},i_{3}}\left|\hat{h}\left(n\bar{l}_{\gamma_{1}}\eta\right)\right|^{2},
\end{align}
where the sum is over $n$-cycles and $\bar{l}_{\gamma}$ is the mean metric length of the edges in $\gamma$. By assumption we have $\bar{l}_{\gamma}=l+o(1)$ as $n\to\infty$ and hence
\begin{align}
    \frac{1}{N^{n}}\sum_{\gamma_{1}}l_{i_{1},i_{2}}l_{i_{2},i_{3}}\left|\hat{h}\left(n\bar{l}_{\gamma_{1}}\eta\right)\right|^{2}&=\left[1+o(1)\right]\left|\hat{h}\left(nl\eta\right)\right|^{2}\frac{1}{N^{n}}\sum_{\gamma_{1}}l_{i_{1},i_{2}}l_{i_{2},i_{3}}\\
    &=\left[1+o(1)\right]l^{2}\left|\hat{h}\left(nl\eta\right)\right|^{2}.
\end{align}
We thus obtain the same result as in the equal length case but with an $o(1)$ error depending on the rate at which $\bar{l}_{\gamma}$ converges to $l$.

\section{Further Remarks and Numerics}\label{s:Remarks}

\subsection{The regime $\eta\ll N^{-1/2}$}
In the region $N^{-1}\ll \eta\ll N^{-1/2}$, the dominant contribution is no longer from cycles because the typical closed walk has at least two self-intersections. In model 2 this can be seen from the fact that falling factorial becomes exponentially small:
\begin{align}
    (N)_{v}&=N^{v}\prod_{i=0}^{v-1}\left(1-\frac{i}{N}\right)\leq N^{v}e^{-\frac{v(v-1)}{2N}}.
\end{align}
To determine the dominant contribution, we consider the identity
\begin{align}
    1&=\sum_{v=1}^{n}\frac{(N)_{v}S(n,v)}{N^{n}},
\end{align}
where we recall that $S(n,v)$ are the Stirling numbers of the second kind. If $n^{2}\ll N$, then the sum over $[1,n-1]$ is negligible and the main term corresponds to $v=n$. In general, we expect that the sum over $[1,n-Cn^{2}/N]$ is negligible for some $C>0$. To see this, we use the bound
\begin{align}
    S(n,k)&\leq\frac{1}{2}\begin{pmatrix}n\\v\end{pmatrix}v^{n-v}\leq\frac{n^{n-v}}{2\sqrt{2\pi v(1-v/n)}}\cdot\frac{(v/n)^{n-2v}}{(1-v/n)^{n-v}},
\end{align}
to obtain
\begin{align}
    \sum_{v=1}^{n-Cn^{2}/N}\frac{(N)_{v}S(n,v)}{N^{n}}&\leq C\sum_{v=1}^{n-Cn^{2}/N}\sqrt{\frac{1}{v(1-v/n)}}\cdot\frac{n^{n-v}(v/n)^{n-2v}}{N^{n-v}(1-v/n)^{n-v}}\cdot e^{-\frac{v^{2}}{2N}}\\
    &=C\sum_{x=1/n}^{1-Cn/N}\sqrt{\frac{n}{x(1-x)}}e^{-nf(x)},
\end{align}
where $x=v/n$ and
\begin{align}
    f(x)&=\frac{n}{2N}x^{2}-(1-x)\log\frac{n}{N}-(1-2x)\log x+(1-x)\log(1-x).
\end{align}
For large enough $C>0$, we can show that $f(x)>cn/N$ in $[0,1-Cn/N]$ and hence the sum is exponentially small. Thus the dominant contribution comes from closed paths of length $n$ on $v$ distinct vertices such that
\begin{align}
    \frac{n-v}{v}&\leq\frac{Cn}{N}\ll1.
\end{align}

\subsection{Single Random Unitary} \label{s:single}
Model 2 is maximally random in the sense that each vertex has an independent random unitary. In the opposite extreme, for regular graphs we can put the same random unitary at each vertex, after conjugating by the permutation matrix corresponding to the choice of ordering $f_{v}$. In the case of the complete graph with self-loops, we can choose $f_{v}=1$ identically so that there is no conjugation by permutation matrices. In this case the amplitude of a closed path $\gamma$ is
\begin{align}
    A_{\gamma}&=U_{i_{n},i_{2}}U_{i_{1},i_{3}}\cdots U_{i_{n-2},i_{n}}U_{i_{n-1},i_{1}}.
\end{align}
Consider the equilateral case $l_{e}=l$, in which we can sum over all vertices to obtain $(\tr U^{k})^{2}$ if $n=2k$ and $\tr U^{2k+1}$ if $n=2k+1$. It is known that
\begin{align}
    \mbb{E}(\tr U^{n})^{p}(\overline{\tr U^{m}})^{q}&\propto\delta_{np,mq},
\end{align}
and
\begin{align}
    \mbb{E}|\tr U^{n}|^{2p}&=[1+o(1)]p!(n\wedge N)^{p}.
\end{align}
Thus we find
\begin{align}
    \mc{V}_{\lambda,\eta}(h)&\simeq2l^{2}\eta^{2}\sum_{n=1}^{\infty}\left[2(n\wedge N)^{2}\left|\hat{h}\left(2nl\eta\right)\right|^{2}+((2n+1)\wedge N)\left|\hat{h}\left((2n+1)l\eta\right)\right|^{2}\right]\\
    &\gtrsim 4\sum_{n=1}^{\infty}n^{2}l^{2}\eta^{2}\left|\hat{h}\left(2nl\eta\right)\right|^{2}\\
    &=\Omega(\eta^{-1}),
\end{align}   
and the variance diverges linearly in $\eta^{-1}$. This is the behaviour one would expect if the zeros behaved as independent random variables (Poisson statistics).

\subsection{Numerical investigation into LSS behaviour with increasing randomness}
Here we interpolate between the model in the above Section \ref{s:single} and Model 2. We use a complete graph with self-loops and equal bond lengths. A subset of \(p\) vertices is assigned independent Haar-distributed matrices, while the remaining \(N-p\) vertices are assigned a common randomly chosen matrix \(U_0\). This makes \(p=0\)  the single random unitary regime of Section \ref{s:single}, whereas \(p=N\) is the maximally random endpoint in Model 2.  We note that the LSS grows with $N$ when $p=0$ and asymptotes toward the RMT predicted, see Figure \ref{LSSvsRandomness}.
\begin{figure}[h]
\begin{center}
 \includegraphics[width = 0.8 \textwidth]{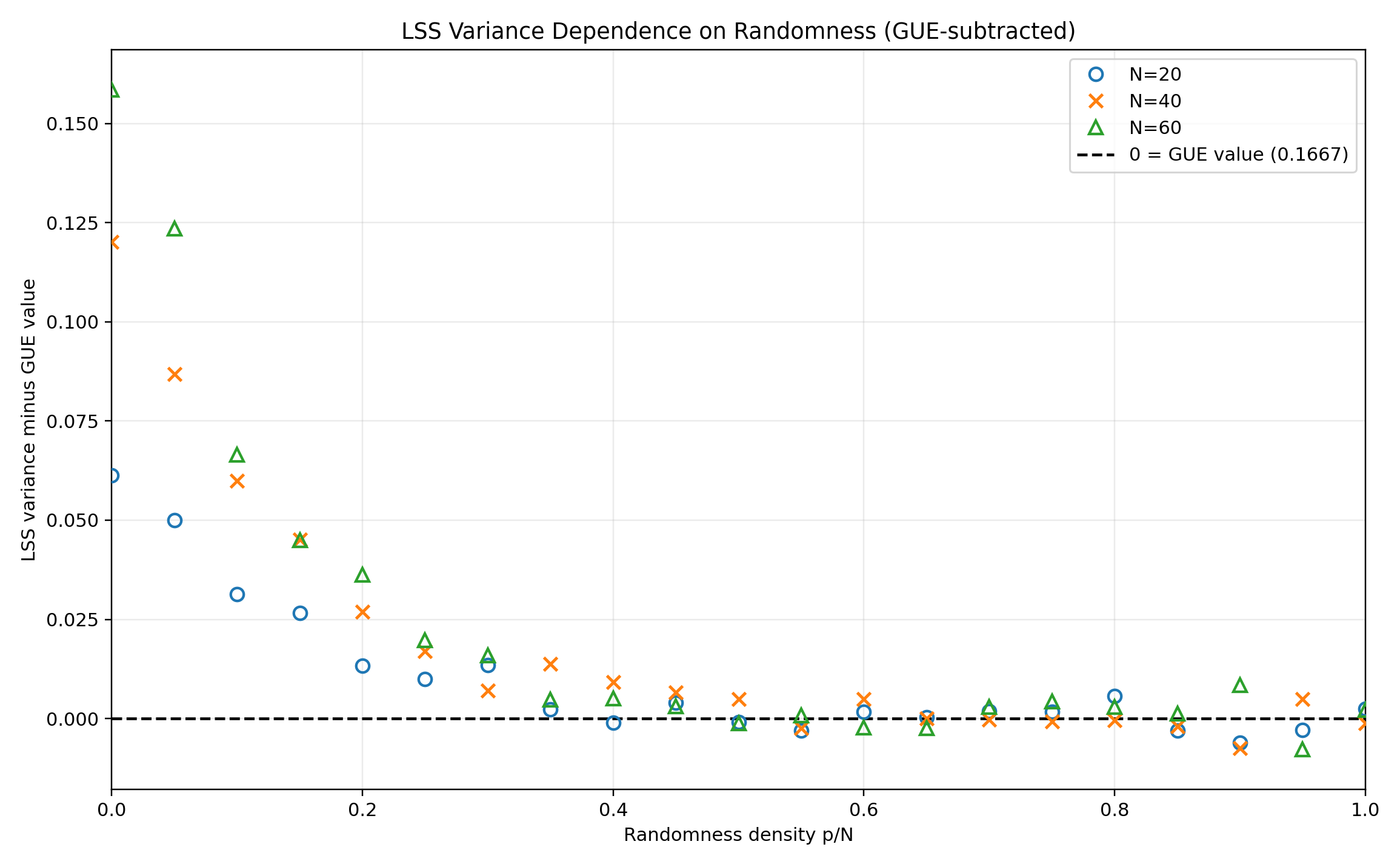} 
 \caption{LSS variance with GUE predicted variance of $1/6$ subtracted out vs proportion of vertices that have iid Haar-distributed unitary scattering matrices}\label{LSSvsRandomness}
 \end{center}
\end{figure}
The key simplification in the equilateral case is that the evolution matrix has the form \(U(k)=e^{ikl}S\), with \(S\) independent of \(k\). Hence, if \(e^{i\theta_j}\) are the eigenvalues of \(S\), then the corresponding \(k\)-eigenvalues are given exactly by \(k=(2\pi m-\theta_j)/l\) for integers \(m\), so the spectral problem is reduced to a single diagonalisation of the \(N^2\times N^2\) scattering matrix \(S\). For each value of $p/N$ we did 5000 independent samples. For each sample we generated 
the required Haar matrices, built the scattering matrix \(S\), diagonalised it, enumerated all roots in the window \(|k-\lambda|\leq u_{\mathrm{cut}}\eta\), and evaluated the truncated linear statistic
\begin{align}
	L_{\lambda,\eta}^{(u_{\mathrm{cut}})}(h)&=\sum_{|k_j-\lambda|\leq u_{\mathrm{cut}}\eta}h\!\left(W(k_j-\lambda)\right),
\end{align}
with
\begin{align}
	h(u)&=\left(\frac{\sin u}{u}\right)^2,\quad \eta=N^{-0.49},\quad  u_{cut}= 20, \quad \lambda = 50.
\end{align}

%A heuristic estimate for a crossover at the order of $p/N\eta$ can be gleaned from the trace-formula proof. Since \(\hat h\) is compactly supported, only periodic orbits of topological length of order \(W\) contribute to the variance. In the mixed model, each visit of an orbit to a vertex lands in the independently sampled subset with probability \(p/N\), so an orbit of length \(W\) encounters on average \(pW/N\) independent vertices. Thus \(pW/N\) may be viewed as the effective randomness seen by that orbit. When \(pW/N\ll 1\), most relevant orbits avoid the independent subset and the single-\(U\) reduction in the variance similar to Section \ref{s:single} is prevalent. The appearance of \((\mathrm{tr}\,U^k)^2\) and \(\mathrm{tr}\,U^{2k+1}\) lead to the larger variance of the less random regime. When \(Wp/N\gtrsim 1\), a typical relevant orbit encounters one or more independent vertices, these high powers of a single matrix are reduced, and the variance tends rapidly to the GUE prediction. 

\subsection{Random quantum graphs corresponding to a self-adjoint laplacian}
In the introduction we defined a quantum graph by assigning unitary matrix-valued functions to each vertex. Not all such functions correspond to a self-adjoint Laplacian. The self-adjoint extensions of the Laplacian can be characterised by vertex conditions relating a function and its derivatives evaluated at each vertex. Namely, for each vertex $v$ of degree $d$ there are $d\times d$ matrices $A_{v},\,B_{v}$ such that
\begin{align}
    A_{v}F+B_{v}F'&=0,\label{eq:vertexConditions}
\end{align}
where $F=(f_{1}(v),...,f_{d}(v))^{T}$ and $F'=(f'_{1}(v),...,f'_{d}(v))^{T}$ for a function $f\in\bigoplus_{e\in E} L^{2}([0,l_{e}])$. We denote by $\mc{D}$ set of $H^{2}$ functions satisfying \eqref{eq:vertexConditions}. The following well-known result establishes a necessary and sufficient condition on $(A_{v},B_{v})$ for a self-adjoint Laplacian.
\begin{theorem}[Kostrykin--Schrader \cite{kostrykin_schrader1999}]
The Laplacian is self-adjoint on $\mc{D}$ iff $A_{v}B_{v}^{}$ is self-adjoint. In this case the $S$-matrix at $v$ is given by
\begin{align}
    S^{(v)}(k)&=-(A_{v}+ikB_{v})^{-1}(A_{v}-ikB_{v}).
\end{align}
\end{theorem}

In light of this result, it is natural to construct ensembles of random quantum graphs by taking $A_{v}$ and $B_{v}$ to be random matrices such that $A_{v}B_{v}^{}$ is self-adjoint. Restricting to invertible $B$, we can write
\begin{align}
    S^{(v)}(k)&=-(L_{v}+ik)^{-1}(L_{v}-ik),
\end{align}
where $L_{v}:=B_{v}^{-1}A_{v}$ is self-adjoint. In this case we can equivalently sample $L_{v}$ from an ensemble of $d\times d$ random self-adjoint matrices. If the distribution of $L_{v}$ is invariant under unitary conjugation, then so is the distribution of $S_{v}(k)$, which implies that the eigenvector matrix of $S_{v}(E)$ is Haar-distributed and independent of the eigenvalues.

\section{Acknowledgements}
The authors are grateful to Evans Harrell for helpful discussions. A.M. is supported by the Royal Society grant number URF$\backslash$R$\backslash$221017 and M. O. by the Royal Society grant number RF/ERE/210051.

\bibliographystyle{abbrv}

\end{document}